\documentclass{article}

\usepackage{amsmath,amsthm}
\usepackage{fullpage}

\usepackage{float}
\usepackage{multicol}
\usepackage{algorithmic}
\usepackage[ruled,vlined,commentsnumbered,titlenotnumbered]{algorithm2e}
\usepackage{verbatim}
\usepackage{latexsym}
\usepackage{amsmath,amssymb,amsthm}
\usepackage{epsfig,ifpdf,graphics}
\usepackage{enumerate}

\newcommand{\Ot}{\tilde{O}}

\newcommand{\color}{\mbox{color}}

\newtheorem{theorem}{Theorem}
\newtheorem{lemma}{Lemma}
\newtheorem{corollary}{Corollary}
\newtheorem{claim}[theorem]{Claim}

\def \eps {\varepsilon}
\def \poly {\textrm{~poly}}

\newcommand{\etal}{{\em et al.\ }}
\newenvironment{reminder}[1]{\smallskip
\noindent {\bf Reminder of #1  }\em}{}

\title{Minimum Weight Cycles and Triangles:\\ Equivalences and Algorithms}

\author{
Liam Roditty \\ Bar Ilan University \and
Virginia Vassilevska Williams  \\ UC Berkeley
}

\begin{document}

\date{}
\maketitle
\thispagestyle{empty}

\begin{abstract}
We consider the fundamental algorithmic problem of finding a cycle of minimum weight in a weighted graph.
In particular, we show that the minimum weight cycle problem in an undirected $n$-node graph with edge weights in $\{1,\ldots,M\}$ or in a directed
 $n$-node graph with edge weights in $\{-M,\ldots , M\}$ and no negative cycles can be efficiently reduced to finding a minimum weight {\em triangle} in an $\Theta(n)-$node \emph{undirected} graph with weights in $\{1,\ldots,O(M)\}$. Roughly speaking, our reductions imply the following surprising phenomenon: a minimum cycle with an arbitrary number of weighted edges can be ``encoded'' using only \emph{three} edges within roughly the same weight interval!

This resolves a longstanding open problem posed in a seminal work by Itai and Rodeh [SIAM J. Computing 1978 and STOC'77] on minimum cycle in unweighted graphs.

A direct consequence of our efficient reductions are  $\tilde{O}(Mn^{\omega})\leq \tilde{O}(Mn^{2.376})$-time algorithms using fast matrix multiplication (FMM) for finding a minimum weight cycle in both undirected graphs with integral weights from the interval $[1,M]$ and directed graphs with integral weights from the interval $[-M,M]$. The latter seems to reveal a strong separation between the all pairs shortest paths (APSP) problem and the minimum weight cycle problem in directed graphs as the fastest known APSP algorithm has a running time of $O(M^{0.681}n^{2.575})$ by Zwick [J. ACM 2002].

In contrast, when only combinatorial algorithms are allowed (that is, without FMM) the only known solution to minimum weight cycle is by computing APSP. Interestingly, any separation between the two problems in this case would be an amazing breakthrough
as by a recent paper by Vassilevska W. and Williams [FOCS'10], any $O(n^{3-\eps})$-time algorithm ($\eps>0$) for minimum weight cycle immediately implies a $O(n^{3-\delta})$-time algorithm ($\delta>0$) for APSP.

%
%
\end{abstract}

\newpage
\setcounter{page}{1}

\section{Introduction}

We consider the algorithmic problem of finding a minimum weight cycle (i.e., weighted girth) in weighted directed and undirected graphs. Surprisingly, although the problem is very fundamental, the state of the art for it dates back to a seminal paper by Itai and Rodeh~\cite{Clique1}, first presented in STOC'77, that deals only with the \emph{unweighted} variant of the problem. Itai and Rodeh presented an $O(n^\omega)$-time algorithm for an $n$-node unweighted undirected graph and an $O(n^\omega \log n)$-time algorithm for an $n$-node unweighted directed graph.
(Here $\omega$ is the exponent of square matrix multiplication over a ring, and $\omega<2.376$~\cite{cw90}.)
In the same paper, Itai and Rodeh posed the question whether similar results exist for weighted graphs. In this paper we provide a positive answer to this longstanding open problem by presenting $\tilde{O}(Mn^\omega)$-time algorithms for directed graphs with integral edge weights in $[-M,M]$ (and no negative cycles) and for undirected graphs with integral edge weights in $[1,M]$.

Our algorithmic results are obtained using new reductions that carefully  combine new algorithmic ideas and special combinatorial properties of the minimum weight cycle. More specifically, we reduce the problem to the problem of finding a minimum weight {\em triangle} in a $\Theta(n)-$node \emph{undirected} graph with weights in $\{1,\ldots,O(M)\}$. This reveals also a surprising phenomenon: a minimum cycle with an arbitrary number of weighted edges can be efficiently ``encoded'' using a cycle of only \emph{three} edges whose weights are roughly within the same interval! Moreover, our results imply a strong \emph{equivalence} between the cycle and triangle problems.

\paragraph{Minimum cycle and APSP.}
Recently, in FOCS'10 Vassilevska W. and Williams~\cite{focsy} showed that the minimum weight cycle problem is equivalent to many other graph and matrix problems for which no truly subcubic ($O(n^{3-\eps})$-time for constant $\eps>0$) algorithms are known. They showed that if there is a truly subcubic algorithm for the minimum weight cycle problem, then many other problems such as the all-pairs-shortest-paths (APSP) problem also have truly subcubic algorithms. Hence, the minimum weight cycle problem has a pivotal role in understanding the complexity of many fundamental polynomial problems in a similar spirit to the role of 3SAT for NP-hard problems.

It is interesting to compare the minimum cycle problem with APSP. In directed graphs, the minimum weight cycle can be computed easily by computing APSP. Given the distance $d[u,v]$ between all pairs of vertices $u,v$, the weight of the minimum cycle is
$\min_{u,v} w(u,v)+d[v,u]$. Hence, we can compute the minimum weight cycle in cubic time using Floyd-Warshall's APSP algorithm~\cite{floyd,warshall} (or Pettie's $O(mn+ n^2\log \log n)$ time algorithm~\cite{Pettie04} if the graph is sparse). If the edge weights are integers in $[-M,M]$, we can use Zwick's~\cite{zwickbridge} $O(M^{0.681}n^{2.575})$ time algorithm to obtain an algorithm for minimum cycle with the same runtime.
Improving Zwick's running time and in particular obtaining an $\tilde{O}(Mn^\omega)$ running time for APSP in directed graphs
 is one of today's frontier questions in graph algorithms.
 Our new $\tilde{O}(Mn^\omega)$-time algorithm for minimum cycle in directed graphs shows that it is not really necessary to compute all pairs shortest paths in order to compute the minimum weight cycle in directed graphs. This
  seems to reveal a strong separation between APSP and the minimum cycle problem in directed graphs.

The minimum cycle problem in undirected graphs differs from the problem in directed graphs in that the reduction to APSP no longer works: an edge $(u,v)$ might also be the shortest path from $v$ to $u$, and  $\min_{u,v} w(u,v)+d[v,u]$ might be $2w(u,v)$ and not the weighted girth of the graph. This represents a nontrivial hurdle. Nevertheless, in this paper we show how to overcome this hurdle and obtain an $\tilde{O}(Mn^\omega)$ time algorithm for undirected graphs with integer weights in $[1,M]$. This matches the runtime of the best APSP algorithm in such graphs: In a paper first presented at STOC'92, Seidel~\cite{Seidel} showed that APSP in undirected and unweighted $n$-node graphs can be solved in $\tilde{O}(n^\omega)$ time. In FOCS'99, Shoshan and Zwick~\cite{sz99} (following Galil and Margalit~\cite{GM97}) showed that APSP in undirected $n$-node graphs with integer edge weights in $[0,M]$ can be solved in $\tilde{O}(Mn^\omega)$ time, thus extending Seidel's running time to weighted undirected graphs.

%

\paragraph{Our results: reductions and algorithms.}
We develop our algorithms by first obtaining extremely efficient reductions from the minimum weight cycle problem to the minimum weight triangle problem which preserve the interval in which the weights lie, within a constant factor.


\noindent \emph{Undirected graphs.}
For undirected graphs our results are as follows.
\begin{theorem}
Let $G(V,E,w)$ be an undirected graph with $w:E\rightarrow \{1,\ldots, M\}$ and let $C$ be a minimum cycle in $G$.
There is an $O(n^2 (\log nM) \log n)$ time deterministic algorithm that computes a cycle $\hat{C}$ and
constructs $O(\log n)$ graphs $G'_1,\ldots,G'_k$ on $\Theta(n)$ nodes and edge weights in $\{1,\ldots,O(M)\}$
such that either $w(\hat{C})=w(C)$ or the minimum out of all weights of triangles in the graphs $G'_i$
  is exactly $w(C)$.
%
\label{thm:undir}
\end{theorem}

Since a minimum weight triangle in a graph with weights bounded by $O(M)$ can be found via a single distance product computation in $\tilde{O}(Mn^\omega)$ time~\cite{agm97,Yuval}, we obtain the following corollary.

\begin{corollary}
A minimum weight cycle in an $n$-node undirected graph with integer edge weights in $[1,M]$ can be found in $\tilde{O}(Mn^\omega)$ time.
\end{corollary}

\noindent \emph{Directed graphs.} Our reduction for undirected graphs relies on the fact that distances are symmetric. It is unlikely that it is possible to modify the reduction so that it works for directed graphs as well. Hence, for directed graphs new ideas are required. The reduction to minimum triangle is not as efficient, however, the resulting algorithm for minimum cycle in directed graphs has
the same running time as the one for undirected graphs with nonnegative weights. Our approach for directed graphs can be combined with our approach for undirected graphs to yield an efficient algorithm also for {\em mixed} graphs,  that is, graphs which contain both directed and undirected edges. The approach works, provided the weights of the mixed graph are nonnegative.

When negative edge weights are allowed, a negative cycle may exist. Finding a minimum weight cycle when its weight is negative is an NP-hard problem,
as it solves Hamiltonian cycle. When negative weights are allowed, the minimum cycle problem in the absence of negative cycles is in P for both directed and undirected graphs, but is NP-hard for mixed graphs~\cite{papamix}.
Our techniques for directed graphs are strong enough to support negative edge weights within the same running time as when the weights are nonnegative. This is extremely interesting, as the typical way to reduce the general problem to the nonnegative weights problem involves computing node {\em potentials} (see e.g.~\cite{dirsssp}). These potentials however typically increase the magnitude of the weights to even $\Omega(Mn)$, which would be bad if our goal is to use algorithms that have exponential dependence on the bit representation of the weights, such as $\Ot(Mn^\omega)$. We circumvent the potential approach by focusing on the general problem directly.

We obtain:

\begin{theorem}
Let $G(V,E,w)$ be a directed graph on $n$ nodes, $w:E\rightarrow \{-M,\ldots, M\}$.
In $\tilde{O}(Mn^\omega)$
time one can
construct $O(\log n)$ graphs $G'_1,\ldots,G'_k$ on $\Theta(n)$ nodes and edge weights in $\{1,\ldots,O(M)\}$ so that
the minimum out of all weights of triangles in the graphs $G'_i$
  is exactly the weighted girth of $G$.
\label{thm:dir}
\end{theorem}

%
%
\paragraph{Our results: equivalences.}
Vassilevska W. and Williams~\cite{focsy} showed that the minimum triangle and minimum cycle problems are equivalent, under subcubic reductions.
Their reduction from minimum triangle to minimum cycle only increased the number of nodes and the size of the edge weights by a constant factor.
 However, their reduction from minimum cycle to minimum triangle was not tight; it only proved that an $O(n^{3-\eps})$ algorithm for minimum triangle would imply an $O(n^{3-\eps/3})$ algorithm for minimum cycle. Our reductions, on the other hand, imply a much stronger equivalence between the two problems. This equivalence is especially strong for undirected graphs with integral weights from the range $[1,M]$.

\begin{corollary}
If there is a $T(n,M)$ time algorithm for the minimum cycle problem in undirected graphs with integral edge weights in $[1,M]$, then there is a $T(O(n),O(M))+O(n^2)$ time algorithm for the minimum triangle problem in such graphs.
Conversely, if there is a $T(n,M)$ time algorithm for the minimum triangle problem in undirected graphs with integral edge weights in $[1,M]$, then there is an $O(T(O(n),O(M))\log n+ n^2\log n\log nM)$ time algorithm for the minimum cycle problem in such graphs.
\label{cor:equiv}
\end{corollary}

A natural question is whether the triangle problem is special. Do similar reductions exist between minimum cycle and minimum $k$-cycle for constant $k>3$? We answer this in the affirmative.

\begin{theorem}
Let $k$ be any fixed constant.
Let $G(V,E,w)$ be a graph on $n$ nodes, $w:E\rightarrow \{1,\ldots, M\}$.
One can
construct $O(\log n)$ undirected graphs $G'_1,\ldots,G'_\ell$ on $\Theta(kn)$ nodes and edge weights in $\{1,\ldots,O(M)\}$ so that
the minimum out of all weights of $k$-cycles in the graphs $G'_i$
  is exactly the weighted girth of $G$. Moreover, given the minimum weight $k$-cycle of the graphs $G'_i$, one can find a minimum weight cycle of $G$ in $O(n)$ additional time.
If $G$ is directed, the reduction runs in $\tilde{O}(Mn^\omega)$
time, and if $G$ is undirected, it runs in $\tilde{O}(n^2\log nM)$ time.
\label{thm:equiv2}
\end{theorem}

\paragraph{Our results: approximation.}
Another approach to gain efficiency for problems with seemingly no subcubic time exact algorithms
 has been to develop fast approximation algorithms (see~\cite{zwickbridge,aingworth,almostshortest} in the context of shortest paths).
 Lundell and Lingas~\cite{ll09} gave two approximation algorithms for the girth problem: an $\tilde{O}(n^{1.5})$ time $8/3$-approximation
 for undirected unweighted graphs, and an $O(n^2 (\log n)\log nM)$ time $2$-approximation
for undirected graphs with integer weights in the range $\{1,\ldots, M\}$. Very recently, Roditty and Tov~\cite{liamsoda} improved the approximation factor to $4/3$-approximation for the weighted case while keeping the running time unchanged.
Due to Zwick's~\cite{zwickbridge} $\tilde{O}(n^\omega/\eps \log (M/\eps))$ time $(1+\eps)$-approximation for APSP and the simple reduction from minimum weight cycle in directed graphs to APSP, the girth of a directed graph admits an $(1+\eps)$-approximation in $\tilde{O}(n^\omega/\eps \log (M/\eps))$ time.
Our reduction from Theorem~\ref{thm:undir} implies the same result for undirected graphs with nonnegative weights as well, following up on the work of Roditty and Tov from SODA'11~\cite{liamsoda}.

\begin{theorem}
There is an $\tilde{O}(n^\omega/\eps \log (M/\eps))$ time $(1+\eps)$-approximation algorithm for the minimum cycle problem in undirected graphs with integral weights in $[1, M]$.
\end{theorem}




%
%
%
%
%

\section{Preliminaries}
Let $G(V,E,w)$ be a weighted graph, where $V$ is its set of {\em vertices} or {\em nodes} (we use these terms interchangeably), $E\subseteq V\times V$ is its set of edges, and $w:E\rightarrow \{1,\ldots,M\}$ is a weight function. The function $w(\cdot,\cdot)$ can be extended to the entire $V\times V$ by setting $w(u,v)=\infty$ for every $(u,v)\notin E$. Unless otherwise noted, $n$ refers to the number of nodes in the graph.

An edge can be directed or undirected. An \emph{undirected} graph is a graph with undirected edges only and a directed graph is a graph with directed edges only.
A \emph{mixed } graph is a graph that may have both directed and undirected edges.
All graphs considered in this paper are {\em simple}. A graph is simple if it does not contain self loops or multiple copies of the same edge.
In a simple mixed graph, a node pair $x,y$ cannot be connected by both a directed and an undirected edge. In both directed and mixed simple graphs,
two directed edges $(x,y)$ and $(y,x)$ in opposite directions are allowed since they are considered distinct.

We define a cycle $C$ in a graph $G(V,E,w)$ to be an ordered set of vertices $\{v_1,v_2,\ldots, v_\ell\}$, such that
$(v_i,v_{i+1})\in E$ for every $i < \ell$ and $(v_\ell,v_1)\in E$. Let $w(C)$ be the sum of the weights of the edges of $C$ and let $w_{\max}(C)$ be the weight of the heaviest edge. We denote with $d_C[v_i,v_j]$ the weight of the path that traverses the cycle from $v_i$ to $v_j$ by passing from $v_i$ to $v_{i+1}$ and so on. In the case that $j<i$ we traverse from $v_\ell$ to $v_1$ and continue until we reach $v_i$. Let $n(C)$ denote the number of vertices/edges in $C$.
A cycle $C$ is {\em simple} if no node or edge appears twice in $C$. With this definition, an undirected graph cannot have a simple cycle on $2$ nodes, where as directed
and mixed
graphs can, provided the two cycle edges are in opposite directions and hence distinct.

\section{Our approach}\label{s-approach}


Our reductions are based on a combinatorial property of cycles in weighted directed, undirected and mixed graphs that might be of independent interest. This property is extremely useful as it shows that crucial portions of the minimum weight cycle are shortest paths.
We present this property in the following lemma.



\begin{lemma}[Critical edge] Let $G(V,E,w)$ be a weighted graph, where $w: E \rightarrow\mathbb{R}$, and assume that $G$ does not contain a negative cycle. Let $C=\{v_1,v_2,\ldots, v_\ell\}$ be a cycle in $G$ of weight $w(C)\geq 0$ and let $s\in C$. There exists an edge $(v_i,v_{i+1})$ on $C$ such that $\lceil w(C)/2\rceil-w(v_i,v_{i+1}) \leq d_C[s,v_i]\leq \lfloor w(C)/2\rfloor$
and $\lceil w(C)/2\rceil-w(v_i,v_{i+1}) \leq d_C[v_{i+1},s]\leq \lfloor w(C)/2\rfloor$. Furthermore, if $C$ is a minimum weight cycle in $G$ then $d[s,v_i]=d_C[s,v_i]$ and $d[v_{i+1},s]=d_C[v_{i+1},s]$.
\label{lemma:middle}
\end{lemma}
\begin{proof}
We can assume, wlog, that $s=v_1$. We start to traverse along $C$ from $v_1$ until we reach the first edge $(v_i,v_{i+1})$ that satisfies  $d_C[v_1,v_i]\leq \lfloor w(C)/2\rfloor$ and $d_C[v_{1},v_{i}]+w(v_i,v_{i+1})\geq \lceil w(C)/2\rceil$. Since $d_C[v_1,v_1]=0\leq \lfloor w(C)/2\rfloor$ either we find an edge  $(v_i,v_{i+1})$ that satisfies the requirement, where $i< \ell$ or we reach $v_\ell$ without finding such an edge. In the latter case $d_C[v_1,v_\ell] \leq \lfloor w(C)/2\rfloor$ and since $d_C[v_{1},v_{\ell}]+w(v_\ell,v_{1})=w(C)\geq \lceil w(C)/2\rceil$ the edge $(v_\ell,v_1)$ satisfies the requirement.

It follows immediately from the properties of the edge $(v_i,v_{i+1})$ that $d_C[v_{1},v_{i}]\geq \lceil w(C)/2\rceil - w(v_i,v_{i+1})$ and hence we get that $\lceil w(C)/2\rceil-w(v_i,v_{i+1}) \leq d_C[v_1,v_i]\leq \lfloor w(C)/2\rfloor$ as required.

We now bound $d_C[v_{i+1},v_1]$. We know that $d_C[v_{i+1},v_1] = w(C) - (d_C[v_1,v_i]+w(v_i,v_{i+1}))$. Since  $d_C[v_{1},v_{i}]+w(v_i,v_{i+1})\geq \lceil w(C)/2\rceil$ we get that $d_C[v_{i+1},v_1] \leq \lfloor w(C)/2\rfloor$. Also, since $d_C[v_1,v_i]\leq \lfloor w(C)/2\rfloor$ we get that $d_C[v_{i+1},v_1] \geq \lceil w(C)/2\rceil-w(v_i,v_{i+1})$.

It remains to show that if $C$ is a minimum weight cycle, then $d[v_1,v_i]=d_C[v_1,v_i]$ and $d[v_{i+1},v_1]=d_C[v_{i+1},v_1]$. If $G$ is a directed graph, then it is straightforward to see that the minimality of $C$ implies that $d_C[u,v]=d[u,v]$ for every $u,v\in C$ and in particular $d[v_1,v_i]=d_C[v_1,v_i]$ and $d[v_{i+1},v_1]=d_C[v_{i+1},v_1]$ as required. Thus, we only need to consider the case that $G$ is an undirected
graph. From the first part of the proof we know that
$d_C[v_{i+1},v_{1}]\leq \lfloor w(C)/2\rfloor$. If $d[v_{i+1},v_{1}]<d_C[v_{i+1},v_{1}]$, then let $P$ be the path from $v_{i+1}$ to $v_1$ of weight $d[v_{i+1},v_{1}]$ and let $C_2$ be the portion of $C$ from $v_1$ to $v_{i+1}$. The union of $P$ and $C_2$ is a walk in $G$ whose weight is strictly less than $w(C)$. Furthermore, since $d[v_{i+1},v_{1}]<d_C[v_{i+1},v_{1}]\leq \lfloor w(C)/2\rfloor\leq w(C_2)$, $P$ and $C_2$ must differ by at least one edge and hence $P\cup C_2$ contains some simple cycle of weight less than $w(C)$, a contradiction to the minimality of $C$.
The argument for showing that $d[v_1,v_i]=d_C[v_1,v_i]$ is symmetric.
\end{proof}

Lemma~\ref{lemma:middle} shows that it is possible to decompose every cycle into three portions: a single edge of weight at most $O(M)$ and two pieces whose weight differs by at most $O(M)$, and which are shortest paths if the cycle is of minimum weight. This observation is crucial for our efficient reductions.
Another important piece of Lemma~\ref{lemma:middle} is that \emph{every} vertex on the cycle has a critical edge. This is especially important in the directed graph case.

Armed with Lemma~\ref{lemma:middle} we can describe the general framework of our approach.
Suppose that we have some way to compute a function $D:V \times V \rightarrow \mathbb{R}$ that satisfies:

\begin{itemize}
\item For every $u,v\in V$, $d[u,v] \leq D[u,v]$
\item There exists a vertex $v$ on the minimum cycle $C$ whose critical edge $(x,y)$ endpoints satisfy $D[v,x]=d[v,x]$ and $D[y,v]=d[y,v]$.
\end{itemize}

In Section~\ref{s-undirected} we show how to compute a function $D$ in $O(n^2\log n\log Mn)$ time for undirected graphs with integral weights from $[1,M]$, and in Section~\ref{s-directed} we show how to compute a function $D$ in $\tilde{O}(Mn^\omega)$ time for directed
graphs with integral weights from $[-M,M]$ and no negative cycles. 



Now consider the following (multi-)graph $G'(V',E',w')$ where $V'=V^1\cup V^2$ and $V^1,V^2$ are disjoint copies of $V$. For every $D[a,b]$ which was computed we place an edge between $a^1\in V^1$ and $b^2\in V^2$  and (for directed graphs) also one between $a^2\in V^2$ and $b^1\in V^1$. These edges get weight $D[a,b]$ and correspond to the two large portions of the minimum cycle. Further, for every edge $(a,b)$ in $G$, we add an edge from $a^2\in V^2$ to $b^2\in V^2$ with weight $w(a,b)$, i.e. $V^2$ induces a copy of $G$; these edges correspond to the critical edge of the minimum cycle. In our reduction for directed graphs
we further transform $G'$ into a simple undirected tripartite graph.

 Consider  $v,x,y$ from the second bullet above.
By Lemma~\ref{lemma:middle}, $D[v,x]+w(x,y)+D[y,v]=d_C(v,x)+w(x,y)+d_C(y,v)=w(C).$
Hence $G'$ will contain $\{v^1,x^2,y^2\}$ as a triangle of weight $w(C)$.
Our reductions in the next two sections give transformations which ensure that every triangle in $G'$ corresponds to a simple cycle in $G$ and that $\{v^1,x^2,y^2\}$ is preserved as a triangle. Since the values $D[\cdot,\cdot]$ are upper bounds on the distances, $\{v^1,x^2,y^2\}$ is a minimum weight triangle in $G'$. The graph $G'$ however can have really large weights; $D[\cdot,\cdot]$ can be as large as $Mn$ in general. Thus our transformations also apply a weight reduction technique which reduces all edge weights to the interval $[-O(M),O(M)]$. This technique is different for undirected and directed graphs.

\paragraph{Finding a minimum cycle.}
Our reductions show that the minimum cycle problem can be efficiently reduced to the minimum triangle problem in a different graph with roughly the same number of nodes and weight sizes. Here we briefly
discuss how one can actually find a minimum weight triangle in an $n$-node graph $G(V,E,w)$ with integral edge weights in the interval $[-M,M]$.
With our reductions, this will give algorithms for the minimum cycle problem as well.

Let $A$ be the $n\times n$ adjacency matrix of $G$ defined as $A[i,j]=w(i,j)$ whenever $(i,j)\in E$ and $A[i,j]=\infty$ otherwise.
A well known approach to finding a minimum weight triangle mimics Itai and Rodeh's algorithm for unweighted triangle finding~\cite{Clique1}.
It first computes the distance product of $A$ with itself, $(A\star A)[i,j]=\min_k A[i,k]+A[k,j]$, to find for every pair of nodes $i,j$ the minimum length of a path with at most $2$ edges between them. Then, the weight of a minimum triangle is exactly $$\min_{i,j} A[j,i]+(A\star A)[i,j].$$
Finding the actual triangle takes only $O(n)$ time after one finds $i,j$ minimizing the above expression.
Thus the running time is dominated by the time required for computing $A\star A$.
The algorithm of Alon, Galil and Margalit~\cite{agm97} (following Yuval~\cite{Yuval}) does this in $\tilde{O}(Mn^\omega)$ time, whenever the entries of $A$ are integers in $[-M,M]$. Hence a minimum triangle, can be found in $\tilde{O}(Mn^\omega)$ time.

%
%
%

\section{Minimum weight cycle in undirected graphs with weights in $\{1,\ldots, M\}$}\label{s-undirected}
Let $G(V,E,w)$  be an undirected graph  with integral edge weights from the range $[1,M]$. In this section we show that in $\Ot(n^2\log Mn)$ time we can compute a cycle whose weight is at most twice the weight of the minimum weight cycle and a new undirected graph $G'(V',E',w')$ with integral weights from the range $[-M,M]$ whose minimum triangle if exists corresponds to the minimum weight cycle in $G$, with constant probability.
(To boost the probability of success, we actually create $O(\log n)$ graphs $G'$.)
If $G'$ does not have a triangle then the cycle that we have computed is the minimum weight cycle of $G$. We start by presenting an $\Ot(n^2)$ time algorithm that given an integer $t$ either reports a cycle of length $2t$ or computes all distances up to $t$.
The computed distances are used to form the graph $G'$.%

\paragraph{Cycle or Distance Computation.} The algorithm works in iterations and in each iteration it repeats the same procedure from a new vertex of the graph. This procedure is a simple adaptation of Dijkstra's algorithm. The input in each iteration is a source vertex $s$ and an integer value $t$. The algorithm either reports a cycle of length at most $2t$ or computes the distances from $s$ to every vertex that is within distance $t$ from $s$. Lingas and Lundell~\cite{ll09} used a similar approach in order to compute a $2$-approximation of the minimum weight cycle. Their algorithm, however, either returns a cycle of length at most $2t$ or computes the distances from $s$ to every vertex that is within distance $2t$ from $s$. This small difference between the two algorithms is crucial for our needs.
The algorithm is given in Algorithm~\ref{A-min-cycle}. The running time of the algorithm is $O(n^2 \log n)$. The algorithm repeats the procedure Cycle? $n$ times, each time with a different vertex. Every run of Cycle? takes at most $O(n \log n)$ time since it stops with the first cycle it detects.

In the next Lemma we prove an important property of the algorithm.

\begin{lemma}\label{L-bin-search}
For any integer $t$, Min-Cycle$(G,t)$  either finds a cycle of weight at most $2t$, or computes all distances of length at most $t$.
\end{lemma}
\begin{proof}
A cycle is reported when a vertex $u$ is extracted from the priority queue $Q$ and for one of its edges $(u,v)$ that is being relaxed the value of $d[v]$ before the relaxation is not infinity. As any distance and any distance estimation are at most $t$,  if a cycle is reported it must be of length at most $2t$. If a cycle is not reported, then our algorithm is almost identical to Dijkstra's algorithm. The only difference is that our algorithm relaxes an edge $(u,v)$ when $u$ is extracted from the priority queue if and only if $d[u]+w(u,v) \leq t$, while Dijkstra's algorithm relaxes all edges of $u$ with no restriction. This implies that our algorithm computes all distances that are smaller or equal $t$.
%
\end{proof}

%

\begin{table}[t]
\begin{multicols}{2}

\begin{algorithm}[H]\label{A-min-cycle}
\caption{Min-Cycle($G,t$)}
\ForEach{$s\in V$}{$C'\gets$ Cycle?($G,s,2t$)\;
\lIf{$w(C') < w(C^*)$}{$C^* \gets C'$}
}
\Return $C^*$
\end{algorithm}

\begin{algorithm}[H]\label{A-cycle?}
\caption{Cycle?($G,s,t'$)}
\lForEach{$v\in V$}{$d[v] \gets \infty$\;}
$d[s] = 0$\;
$Q \gets \{ s \}$\;
\While{$Q \neq \emptyset$}
{
$u \gets $ Extract-Min($Q$)\;
Controlled-Relax($u,t'/2$)\;
}

\end{algorithm}

\columnbreak

\begin{algorithm}[H]\label{A-control-relax}

\caption{Controlled-Relax($u,w_u$)}
$(u,v)\gets $ Extract-Min($Q_u$)\;
\While{$d[u]+w(u,v) \leq w_u$}{
    \eIf{$d[v] \neq \infty$}
    {report a cycle and stop\;}
    {Relax($u,v$)\;}
    $(u,v)\gets $ Extract-Min($Q_u$)\;
}
\end{algorithm}

\begin{figure}[H]
\centering \includegraphics[height=0.6cm]{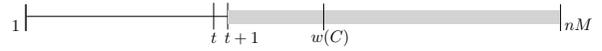}
\vspace{-9pt}\caption{In the gray area Min-Cycle reports a cycle.}\label{F-critical-value}
\end{figure}

\end{multicols}
\end{table}


\paragraph{The reduction to minimum triangle.}
Our goal is to prove Theorem~\ref{thm:undir}. The main part of the proof is describing an algorithm that computes an upper bound for the minimum weight cycle and an instance $G'$ of minimum triangle, such that either the girth of the graph is exactly the upper bound, or with constant probability the minimum triangle weight in $G'$ is the girth of $G$.
%
%
%
%
Below we only present $G'$ as having large weights. Later on, we find a value
$t$ with which we use Lemma~\ref{L-bin-search}, so that $2t$ is a bound on the minimum cycle weight that is tight within $O(M)$. As mentioned in Section~\ref{s-approach}, this value allows us to reduce the weights of $G'$ so that they fall in the interval $[-O(M),O(M)]$.

\begin{reminder}{Theorem~\ref{thm:undir}}
Let $G(V,E,w)$ be an undirected graph with $w:E\rightarrow \{1,\ldots, M\}$ and let $C$ be a minimum cycle in $G$.
There is an $O(n^2 (\log nM) \log n)$ time deterministic algorithm that computes a cycle $\hat{C}$ and
constructs $O(\log n)$ graphs $G'_1,\ldots,G'_k$ on $\Theta(n)$ nodes and edge weights in $\{1,\ldots,O(M)\}$
such that either $w(\hat{C})=w(C)$ or the minimum out of all weights of triangles in the graphs $G'_i$
is exactly $w(C)$.
\end{reminder}

The weight of the minimum cycle is an integer value from the range $[1,nM]$. From Lemma~\ref{L-bin-search} it follows that we can use algorithm Min-Cycle to perform a binary search over this range in order to find the largest value $t\in [1,nM]$ for which Min-Cycle$(G,t)$ does not report a cycle but computes all distances of length at most $t$ (see Figure~\ref{F-critical-value}). This implies that by running Min-Cycle$(G,t+1)$ we obtain a cycle of weight at most $2t+2$. Hence, we only need to show that it is possible to detect the minimum cycle in the case that its weight $w(C)$ is $2t+1$ or less.
Let us first prove some consequences of the fact that Min-Cycle$(G,t)$ does not report a cycle.

\begin{lemma}
Let $C=\{v_1,v_2,\ldots, v_\ell\}$ be a minimum cycle in $G(V,E,w)$.
Suppose that Min-Cycle$(G,t)$ does not report a cycle. There are three \textbf{distinct} vertices $v_i,v_{i+1},v_j\in C$ such that
$d_C[v_j,v_i] + w(v_i,v_{i+1}) > t$ and $d_C[v_{i+1},v_j] + w(v_i,v_{i+1}) > t$.
\label{lemma:tbound}
\end{lemma}

\begin{proof}
Let $(v_i,v_{i+1})$ be the critical edge for $v_1$ given by Lemma~\ref{lemma:middle}. Assume first that $v_1\neq v_i$ and $v_1\neq v_{i+1}$.
If either $d_C[v_1,v_i] + w(v_i,v_{i+1}) \leq t$ or $d_C[v_{i+1},v_1] + w(v_i,v_{i+1}) \leq t$ then the edge $w(v_i,v_{i+1})$ is relaxed. Assume that we are in the case that $d_C[v_1,v_i] + w(v_i,v_{i+1}) \leq t$. Then after $(v_i,v_{i+1})$ is relaxed $d[v_{i+1}]\leq t$. If
$d[v_{i+1}]<\infty$ before the relaxation of  $(v_i,v_{i+1})$  the algorithm stops and reports a cycle. If $d[v_{i+1}]=\infty$ before the relaxation of  $(v_i,v_{i+1})$  then a cycle will be detected as well but only when the edge $(v_{i+2},v_{i+1})$ is relaxed. This edge must be relaxed since $d_C[v_{i+1},v_1]\leq \lfloor w(C)/2\rfloor \leq t$ which implies that $v_{i+2}$ will be extracted and its edge $(v_{i+2},v_{i+1})$ will satisfy the relaxation requirement and will be relaxed. We conclude that if either $d_C[v_1,v_i] + w(v_i,v_{i+1}) \leq t$ or $d_C[v_{i+1},v_1] + w(v_i,v_{i+1}) \leq t$ then Min-Cycle$(G,t)$ must report a cycle, giving a contradiction.

We now turn to the case that either $v_1 = v_i$ or $v_1 = v_{i+1}$. Assume wlog that $v_1=v_i$, that is,  $(v_i,v_{i+1})=(v_1,v_2)$.
From Lemma~\ref{lemma:middle} we know that $w(v_1,v_2)\geq \lceil w(C)/2 \rceil$ and $d_C[v_2,v_1]\leq \lfloor w(C)/2 \rfloor$. We also know that there is at least one additional vertex $v_\ell$ between $v_2$ and $v_1$ on the cycle $C$. We now apply Lemma~\ref{lemma:middle} on the vertex $v_\ell$. It is easy to see that in that case the edge $(v_1,v_2)$ will be the critical edge of $v_\ell$ as well. We now have three different vertices and the rest of this case is identical to the first case.
\end{proof}



As a first attempt, we create the new graph $G'(V',E',w')$ as follows. The vertex set $V'$ contains two copies $V^1$ and $V^2$ of $V$.
For $i=1,2$, let $E^i$ be the set of edges with both endpoints in $V^i$. The set $E^1$ is empty and the set $E^2$ is $E$, that is, $(u^2,v^2)\in E^2$ if and only if $(u,v)\in E$. Let $E^{12}$ be the set of edges with one endpoint in $V^1$ and one endpoint in $V^2$.
Let $u^1 \in V^1$ and $v^2\in V^2$. If the distance between $u$ and $v$ was computed by Min-Cycle$(G,t)$ then we add an edge $(u^1,v^2)$ to $E^{12}$ with weight $d[u,v]$. We show that there is triangle in $G'(V',E',w')$ that corresponds to the minimum cycle of $G$ and has the same weight.

\begin{claim}\label{C-G'-first-try}
Let $C=\{v_1,v_2,\ldots, v_\ell\}$ be a minimum cycle in $G(V,E,w)$. Assume that $w(C)\leq 2t+1$.
There exists a triangle in $G'(V',E',w')$ on vertices of $C$ of weight $w(C)$.
\end{claim}
\begin{proof}
Without loss of generality, let $v_1$ be the vertex $v_j$ from Lemma~\ref{lemma:tbound}, and let $v_i$ and $v_{i+1}$ be the other two vertices. From Lemma~\ref{lemma:tbound} we know that all three vertices are distinct and that $d_C[v_1,v_i] + w(v_i,v_{i+1}) > t$ and $d_C[v_{i+1},v_1] + w(v_i,v_{i+1}) > t$. Combining this with the fact that $C$ is a minimum cycle and $w(C)\leq 2t+1$ we get that $d[v_1,v_i]=d_C[v_1,v_i] \leq t$ and $d[v_{i+1},v_1] = d_C[v_{i+1},v_1]  \leq t$. When Cycle? is run from $v_1$ it computes $d[v_1,v_i]$ and $d[v_1,v_{i+1}]$.
Hence, there must be a triangle of weight $w(C)$ in $G'$ on the vertices $v_1^1$,$v_i^2$ and $v_{i+1}^2$.
\end{proof}

\begin{figure}[t]\label{F-G-to-G'}
\centering \includegraphics[height=4.2cm]{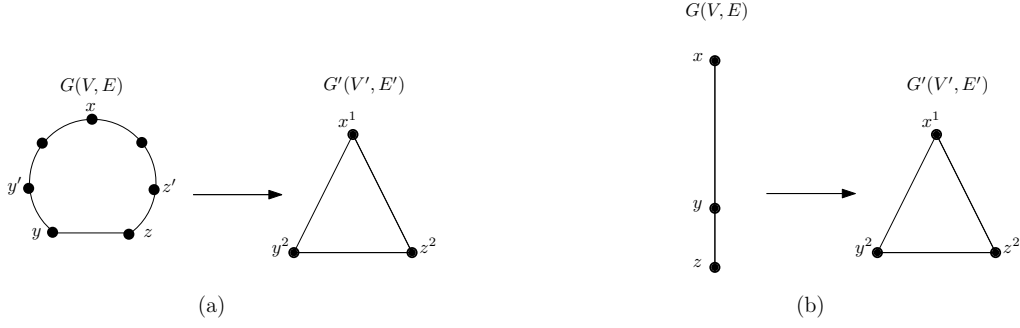}
\caption{(a) A minimum cycle in $G$ that is transformed into a triangle in $G'$. (b) A simple path in $G$ that is transformed into a triangle in $G'$.}
\end{figure}

The claim above shows only one direction, that is, if there is a minimum cycle $C$ of weight at most $2t+1$ in $G$ then there is a corresponding triangle in $G'$ on vertices $y^2,z^2\in V^2$ and $x^1\in V^1$, that correspond to vertices of $C$ with the same weight. This situation is depicted in Figure~\ref{F-G-to-G'}(a). To complete the reduction we must show that there are no false positives: triangles in $G'$ of smaller weight which do not correspond to a minimum cycle of $G$. Unfortunately, this is not the case and $G'$ may have such false triangles with smaller weight. This situation is depicted in Figure~\ref{F-G-to-G'}(b). Let $x,y,z\in V$. If there is a shortest path of length at most $t$ from $x$ to $z$ whose last edge is $(y,z)$ then there is a triangle in $G'$. To see that notice that there are two different shortest paths one from $x$ to $z$ and one from $x$ to $y$, both of length at most $t$. In such a case the graph $G'$ includes the edges $(x^1,y^2)$ and $(x^1,z^2)$ and together with the edge $(y^2,z^2)$ they form a triangle. Moreover, such a triangle has the same structure as a valid triangle and might be of smaller weight, thus, a triangle detection algorithm cannot distinguish between a valid triangle and a false triangle. In what follows we first show that this is the only situation in which a false triangle is formed and then we show a construction that avoids such false triangles.

In the above pathological case the only reason that the triangle $x^1,y^2,z^2$ did not correspond to a cycle, was because we had two different paths $P_1$ and $P_2$ that both start in the same vertex and the last vertex of one of these paths was the vertex right before the last vertex of the other path. In the next lemma we show that this is the {\em only} bad case.

\begin{lemma}
Let $x,y,z\in V$ be three distinct vertices.
Let $P_1=y\rightarrow y'\rightarrow \ldots\rightarrow x$ and $P_2=x\rightarrow \ldots\rightarrow z'\rightarrow z$ be simple shortest paths between $y$ and $x$ and $x$ and $z$ respectively. Let $y'\neq z$ and $z'\neq y$ and let $(z,y)\in E$. Then, $P_1\cup P_2\cup \{(z,y)\}$ contains a simple cycle of weight at most $w(P_1)+w(P_2)+w(z,y)$.\label{lemma:twopaths}
\end{lemma}

\begin{proof}
Let $P^{-1}_1$ be $P_1$ with its edges reversed.
Look at $P^{-1}_1$ and $P_2$.
There are two options. Either one path is a subpath of the other, or there is a node $x'$ such that $x\rightarrow \ldots\rightarrow x'$ is a subpath of both, and $x'$ is followed by $q_1$ in $P^{-1}_1$ and by $q_2\neq q_1$ in $P_2$.

Consider the first case. Wlog, $P_2$ is a subpath of $P^{-1}_1$ (the other inclusion is symmetric).
Since $y'\neq z$, the subpath between $y$ and $z$ on $P_1$ has at least 2 edges, and adding edge $(z,y)$ produces a simple cycle of weight less than the sum of the two original path weights.

Consider the second case when $x',q_1,q_2$ exist as above.
If there is some node between $x'$ and $y$ on $P^{-1}_1$ which also appears in $P_2$ after $x'$, then let $q$ be the first such node. Then no node on $P_2$ between $x'$ and $q$ appears between $x'$ and $q$ in $P^{-1}_1$. The two disjoint simple paths between $x'$ and $q$ form a simple cycle on at least $3$ nodes since $q_1\neq q_2$. The weight of this cycle is less than the sum of the two original path weights.

Finally, suppose no such $q$ exists. Then the subpaths of $P^{-1}_1$ and $P_2$ between $x'$ and $y$ and $x'$ and $z$ share no vertices and hence adding edge $(z,y)$ closes a simple cycle of weight at most $w(P_1)+w(P_2)+w(z,y)$.
\end{proof}

Lemma~\ref{lemma:twopaths} implies that our reduction to minimum triangle will work, provided that we can ensure that for every triangle $x^1,y^2,z^2$ in $G'$, the last node $z'$ before $z$ on the shortest path from $x$ to $z$ in $G$ is distinct from $y$.
To do this, we use the color-coding method from the seminal paper of Alon, Yuster and Zwick~\cite{AYZ97}.
The idea is as follows. Let $\{C_1,C_2\}$ be two distinct colors and suppose we assign to every node of $G$ one of these colors independently and uniformly at random.
Fix four vertices $y,y',z,z'$. The probability that $\color(y')=\color(z')=C_1$ and $\color(y)=\color(z)=C_2$ is $1/2^4=O(1)$.

Now we will modify $G'(V',E',w')$ from before. Recall that $V'=V^1\cup V^2$.
For every node $x$ of $G$ we add a copy $x^1$ to $V^1$, so that $V^1$ is a copy of $V$.
Furthermore, if color$(x)=C_2$ we also add a copy $x^2$ to $V^2$.
We now define the set of edges $E'$. Let $E^{ij}$ be the set of edges between $V^i$ and $V^j$, for $i,j\in \{1,2\}$.
The edge set $E^{11}$ is empty, so $E' = E^{12} \cup E^{22}$.
Let $x,z',z\in V$ such that $(z',z)$ is the last edge of the shortest path from $x$ to $z$. The sets $E^{12}$ and $E^{22}$ are defined as follows:
$$ E^{12} = \{ (x^1,z^2) \mid \color(z')=C_1 \wedge \color(z)=C_2\} \;\;\;\;\;\;\;\;\;\;\;\; E^{22} = \{ (x^2,z^2) \mid \color(x)=C_2 \wedge \color(z)=C_2\}.$$

The weight of an edge $(x^1,z^2)\in E^{12}$ is $d[x,z]$.
The weight of an edge $(x^2,z^2)\in E^{22}$ is $w(x,z)$.
We now prove that $G'$ does not contain false triangles.

\begin{lemma}\label{L-G'-has-no-false-triangle}
If $T=\{x,y,z\}$ is triangle in $G'$ then there exists a simple cycle $C$ in $G$ such that $\{x,y,z\}\subseteq C$ and $w(C) \leq w(T)$.
\end{lemma}
\begin{proof}
Any triangle in $G'$ either have one vertex from $V^1$ and two vertices from $V^2$ or all three vertices from $V^2$. In the latter case the triangle is also in $G$ so we focus in the former case, that is, $T=\{x^1,y^2,z^2\}$ is a triangle in $G'$ such that $x^1\in V^1$ and $y^2,z^2\in V^2$. Let $x,y,z\in V$ be the vertices that correspond to $x^1,y^2$ and $z^2$ in $G$. Let $y'$ ($z'$) be the last vertex before $y$ ($z$) on the shortest path $P_1$ ($P_2$) between $x$ and $y$ ($z$) in $G$. The fact that $(x^1,y^2)\in E^{12}$ and $(x^1,z^2)\in E^{12}$ implies that $\color(y')=\color(z')=C_1$ and $\color(y)=\color(z)=C_2$. Hence we get that $y'\neq z$ and $z'\neq y$.
%
Combining this with the fact that $E^{22}\subseteq E$ we get that the paths $P_1$, $P_2$ and the edge $(y,z)$ satisfy the requirements of Lemma~\ref{lemma:twopaths}, and there is a simple cycle of weight at most $w(P_1)+w(P_2)+w(y,z)=w(T)$ in $G$.
\end{proof}

%
Now that we have shown that $G'$ does not contain false triangles we prove that the minimum weight cycle in $G$ corresponds to a triangle in $G'$.  (This can be viewed as proving Claim~\ref{C-G'-first-try} for the new construction of $G'$).


\begin{claim}\label{C-G'-second-try}
Let $C=\{v_1,v_2,\ldots, v_\ell\}$ be a minimum cycle in $G(V,E,w)$. Assume that $w(C)\leq 2t+1$.
Then there exists a triangle in $G'(V',E',w')$ on vertices of $C$ of weight $w(C)$, with constant probability.
\end{claim}
\begin{proof}
Without loss of generality, let $v_1$ be the vertex $v_j$ from Lemma~\ref{lemma:tbound}, and let $v_i$ and $v_{i+1}$ be the other two vertices. As in the proof of Claim~\ref{C-G'-first-try}, $d[v_1,v_i]=d_C[v_1,v_i]\leq t$ and $d[v_{i+1},v_1]=d_C[v_{i+1},v_1]\leq t$ and these values are computed by  Min-Cycle$(G,t)$ .
The random coloring is successful when $\color(v_{i-1})=C_1$, $\color(v_i)=C_2$, $\color(v_{i+1})=C_2$ and $\color(v_{i+2})=C_1$.
The probability that this happens is $1/2^4 = O(1)$.
The triangle $\{v^1_1,v^2_i,v^2_{i+1}\}$ is in $G'$ exactly when the coloring is successful, and hence $C$ is represented by that triangle in $G'$ with constant probability. The weight of the triangle  $\{v^1_1,v^2_i,v^2_{i+1}\}$ is $d[v_1,v_i]+w(v_i,v_{i+1})+d[v_{i+1}, v_1]= w(C)$.
\end{proof}

\paragraph{Weight reduction.} Currently, the maximum edge weight in $G'$ can be as large as $\Omega(nM)$ as the weights of edges in $E^{12}$ are distances in $G$. To complete the reduction, we show that it is possible to reweight the edges of $G'$ without changing the minimum triangle so that the edge weights will be integers from the range $[-M,M]$.


The key idea is to use Lemma~\ref{lemma:tbound} in two different ways. As we previously mentioned, Lemma~\ref{lemma:tbound} implies that
$d_C[v_j,v_i] \leq t$ and $d_C[v_{i+1},v_j] \leq t$. Moreover, the bounds $d_C[v_j,v_i] + w(v_i,v_{i+1}) > t$ and $d_C[v_{i+1},v_j] + w(v_i,v_{i+1}) > t$ imply that $d_C[v_j,v_i] > t - M$ and $d_C[v_{i+1},v_j] > t - M$. Thus, we can remove from $E^{12}$ every edge of weight strictly more than $t$ and every edge of weight $t-M$ or smaller with no effect on the minimum triangle in $G'$. We now decrease the weights of all the edges that were left in $E^{12}$ by $t$. The weight of every triangle in $G'$ with a node from $V^1$ has decreased by exactly $2t$.
Hence, the minimum triangle out of those with a node in $V^1$ remains the same. The weights of edges in $E^{12}$ are now integers from the interval $[-M,0]$, and the rest of the edge weights are still in $[1,M]$. If the minimum weight triangle in $G'$ now has nodes only from $V^2$, then this triangle was also the minimum weight one in $G'$ before the reweighting, and hence corresponds to a minimum weight cycle, with high probability.
Otherwise, the minimum weight triangle in $G'$ has a node from $V^1$. The minimum out of these triangles was also the minimum one among the triangles with a node in $V^1$ also before the reweighting. Hence it also corresponds to a minimum weight cycle, with high probability.
This completes the description of our construction.

\paragraph{Derandomization.} The reduction can be made deterministic, just as in the color-coding paper of Alon \etal~\cite{AYZ97}, by using a $k$-perfect hash family, a family $F=\{f_1,\ldots,f_{|F|}\}$ of hash functions from $\{1,\ldots, n\}$ to $\{1,\ldots,k\}$ so that for every $V'\subset V$ with $|V'|=k$, there exists some $i$ so that $f_i$ maps the elements of $V'$ to distinct colors. In our case, $k=2$.
By enumerating through the functions of $F$, and using each $f_i$ in place of the random coloring, our reduction runs in $O(n^2(\log nM)\log n + |F| n^2)$ time, provided each $f_i$ can be evaluated in constant time. Our reduction produces $O(|F|)$ instances of minimum triangle.

Schmidt and Siegel~\cite{schmidtsiegel} (following Fredman, Komlos and Szemeredi~\cite{fks84}) gave an explicit construction
of a $k$-perfect family in which each function is specified using $O(k)+2\log\log n$ bits.
For our case of $k=2$, the size of the family is therefore $O(\log^2 n)$. The value of each one of the hash functions on
each specified element of $V$ can be evaluated in $O(1)$ time. Alon, Yuster and Zwick~\cite{AYZ97}, reduced the size of the hash family to $O(\log n)$. Using this family we can derandomize our reduction so that it runs in deterministic $O(n^2(\log nM)\log n)$ time.

\section{Minimum cycle in directed graphs with weights in $\{-M,\ldots, M\}$}\label{s-directed}

In this section we consider directed graphs graphs with possibly negative weights but no negative cycles.
In contrast to the situation in undirected graphs it is relatively easy to reduce the minimum cycle problem in directed graphs to the problem of computing all pairs shortest paths.
If $D$ is the distance matrix of a directed graph then its minimum cycle has weight $\min_{i,j} D[i,j]+w(j,i)$. Hence, using Zwick's APSP algorithm~\cite{zwickbridge} we can compute the minimum cycle in $O(M^{0.681}n^{2.575})$ time.
In this section we show that the minimum cycle problem in directed
 graphs can be reduced to the problem of finding a minimum triangle in an undirected graph. This also implies that the minimum weight cycle in directed
 graphs can be computed in $\tilde{O}(Mn^\omega)$ time.

Similarly to before, our approach will be to compute upper bounds on the distances in the graph so that for some node $s$ on the minimum cycle $C$ and its critical edge $(v_i,v_{i+1})$ we obtain the exact distances $d[s,v_i]=d_C[s,v_i]$ and $d[v_{i+1},s]=d_C[v_{i+1},s]$.

\paragraph{Computing cycle distances.}
The Dijkstra-like approach in the previous section does not work for directed graphs.
It also only applies when the edge weights are nonnegative.
Here we utilize a new approach that allows us to reduce the minimum cycle problem in
directed
graphs with integral weights in the interval $[-M,M]$ to the minimum triangle problem in undirected graphs with weights in $[-M,M]$. Our result is  more general than before. However this comes at a cost: the reduction no longer takes nearly quadratic time, but consumes $\Ot(Mn^\omega)$ time.

Our approach uses the fact that Lemma~\ref{lemma:middle} applies for \emph{every} vertex of a cycle, together with a result by
Yuster and Zwick~\cite{yusterzwick05} obtained by a clever modification of Zwick's APSP algorithm~\cite{zwickbridge} given in Theorem~\ref{thm:yz} below.

\begin{theorem}[Yuster and Zwick '05]\label{thm:yz}
Given an $n$-node directed graph with integral edge weights in the interval $[-M,M]$, in $\tilde{O}(Mn^\omega)$ time one can compute an $n\times n$ matrix $D$ such that the $i,j$ entry of the distance product $D\star D$ contains the distance between $i$ and $j$.
\end{theorem}

The matrix $D$ can contain entries with values as large as $\Omega(Mn)$ and so $D\star D$ is not known to be computable in truly subcubic time, even when $M$ is small. Nevertheless, the theorem applies to general graphs with positive or negative weights. It also gives an $\tilde{O}(Mn^\omega)$ time algorithm for detecting a negative cycle in a graph, and
 is extremely useful in computing minimum cycles.

The Yuster-Zwick algorithm proceeds in stages. In each stage $\ell$, a node subset sample $B_\ell$ is maintained so that each node is in $B_\ell$ with probability at least $\min\{1,9(2/3)^\ell\ln n\}$.
They prove the following lemma.

\begin{lemma}[Yuster and Zwick '05]
For every stage $\ell$ and any node $s\in B_\ell$ and node $v\in V$, the algorithm has estimates $D[s,v]$ and $D[v,s]$, so that if the shortest path from $s$ to $v$ has at most $(3/2)^\ell$ edges then $D[s,v]=d[s,v]$, with high probability. Similarly, if the shortest path from $v$ to $s$ has at most $(3/2)^\ell$ edges then $D[v,s]=d[v,s]$ with high probability.\label{lemma:yz}
\end{lemma}

The Yuster-Zwick algorithm also provides an additional matrix $\Pi$ of {\em predecessors} so that if $k=\Pi[i,j]$, then $k$ is the predecessor of $j$ on a simple path from $i$ to $j$ of weight $D[i,j]$. Similarly, one can obtain a matrix $\Pi'$ of {\em successors} so that if $k=\Pi'[i,j]$, then $k$ is the successor of $i$ on a simple path from $i$ to $j$ of weight $D[i,j]$.

Now, first use the algorithm to check whether the given graph has a negative cycle. If it does not, then
let $C$ be the minimum weight cycle, $w(C)\geq 0$.
Recall that $n(C)$ is the number of vertices/edges on $C$. Let $\ell$ be the minimum value so that $n(C)\leq (3/2)^\ell$. Note that then $n(C)\geq (3/2)^{\ell-1}$.
The probability that a particular node $s$ of $C$ is not in $B_\ell$ is at most $1-(2/3)^\ell(9\ln n)$. The events are independent for all $s$ in $C$, and so the probability that no node of $C$ is in $B_\ell$ is at most

$$(1-(2/3)^\ell (9\ln n))^{n(C)}\leq (1-(2/3)^\ell (9\ln n))^{(3/2)^{\ell-1}}\leq 1/n^6.$$

Thus the probability that some node $s$ of $C$ is in $B_\ell$ is $1-\poly^{-1}(n)$; furthermore by Lemma~\ref{lemma:yz} (with high probability) for all $x\in C$,
the Yuster-Zwick algorithm has computed
 $D[s,x]=d[s,x]$ and $D[x,s]=d[x,s]$, since the number of edges on the respective shortest paths are at most $n(C)\leq (3/2)^\ell$.

In particular, this means that $D[s,v_i]=d[s,v_i]$ and $D[v_{i+1},s]=d[v_{i+1},s]$ for the critical edge $(v_i,v_{i+1})$ for $s$ on $C$ as Lemma~\ref{lemma:middle} applies for every vertex of $C$.
Moreover, since $C$ is a minimum weight cycle with $w(C)\geq 0$, by Lemma~\ref{lemma:middle}, the paths between $s$ and $v_i$ and $v_{i+1}$ and $s$ on $C$ are shortest paths between $s$ and $v_i$ and $v_{i+1}$ and $s$, respectively. Hence, $d_C[s,v_{i}]=d[s,v_{i}]=D[s,v_{i}]$ and $d_C[v_{i+1},s]=d[v_{i+1},s]=D[v_{i+1},s]$, with high probability.

\paragraph{Creating the minimum triangle instance $G'$.}
$G'$ will still be undirected, but unlike the construction for undirected graphs, $G'$ will now be tripartite.
The vertex set $V'$ of $G'$ has partitions $V^1,V^2,V^3$ which are all copies of $V$.

The construction is as follows. For every directed edge $(u,v)$ of $G$, add an edge from $u^2\in V^2$ to $v^3\in V^3$ with weight $w(u,v)$.
Furthermore, for every two nodes $x,y$ so that $D[x,y]<\infty$ add an edge from $x^1\in V^1$ to $y^2\in V^2$ and one from $x^3\in V^3$ to $y^1\in V^1$, each with weight $D[x,y]$.
Hence the edges between $x^1\in V^1$ and $y^2\in V^2$ correspond to directed paths from $x$ to $y$, and the edges between $x^3\in V^3$ and $y^1\in V^1$ correspond to directed paths from $y$ to $x$. Hence any triangle in $G'$ corresponds to a directed closed walk in $G$.
However, any such closed walk must contain a simple cycle of no larger weight:
If the walk is not simple, find a closest pair of copies of a node $v$ on the walk. These copies enclose a simple cycle $C''$ in $G$. Now, either $C''$ has no larger weight than the walk, or removing it from the walk produces a smaller closed walk of negative weight, and hence $G$ contains a negative cycle, which we assumed is not the case.
Since $G'$ contains no false positives, the minimum triangle of $G'$ corresponds exactly to $C$.

\paragraph{Weight reduction.} As in the construction for undirected graphs, the maximum edge weight in $G'$ can be as large as $\Omega(nM)$. Here we give a different way to reduce them to the interval $[-M,M]$.

Let $t$ be a parameter which we will be changing.
Intuitively, our goal will be to set $t$ to something roughly $\lfloor w(C)/2\rfloor$; for our purposes, it will be sufficient for $t$ to be $\leq \lfloor w(C)/2\rfloor$.

Initially, $t=Mn$.
Now, check whether there is a triangle $a^1 \in V^1$, $b^2 \in V^2$, $c^3 \in V^3$ in $G'$ so that $D[a^1,b^2], D[c^3,a^1]\leq t$.
We run a binary search on $t$ in the interval $[0, Mn]$, until we find the smallest $t$ such that there is such a triangle.
Each search can be done using Boolean matrix product: create a matrix $A$ which is $1$ wherever $D$ is $\leq t$ and $0$ otherwise; multiply $A$ by itself and check for a triangle closed by an edge $(b^2,c^3)$, $b^2\in V^2, c^3\in V^3$.
This takes $O(n^\omega \log w(C))$ time.

Let (whp) $\{s^1,v_{i}^2,v_{i+1}^3\}$ be the triangle in $G'$ that corresponds to the minimum cycle $C$ of $G$.
Since $\{s^1,v_{i}^2,v_{i+1}^3\}$  is a valid triangle, and $d_C[s,v_{i}],d_C[v_{i+1},s]\leq \lfloor w(C)/2\rfloor$ by Lemma~\ref{lemma:middle},
then after the completion of the binary search, $t\leq \lfloor w(C)/2\rfloor$. Furthermore,
since $C$ is a minimum cycle and by the definition of $t$,
$w(C) \leq 2t+w(e)$, where $e$ is some edge in $G$, which implies that $w(C)\leq 2t+M$. Hence, $t\leq \lfloor w(C)/2\rfloor$ and $w(C)/2\leq t+M/2$.

Now, take $G'$ and remove every edge $(c^3,a^1)\in V^3\times V^1$ with $D[c^3,a^1]> t+M/2$ and every $(a^1,b^2)\in V^1\times V^2$ with $D[a^1,b^2]>t+M/2$.
If an edge has weight $\leq \lfloor w(C)/2\rfloor\leq t+M/2$, it is not removed. In particular, $(s^1,v_i^2)$ and $(v_{i+1}^3,s^1)$ are still edges, by Lemma~\ref{lemma:middle}.

Remove every $(c^3,a^1)\in V^3\times V^1$ with $D[c^3,a^1]<t-M$ and every $(a^1,b^2)\in V^1\times V^2$ with $D[a^1,b^2]<t-M$.
If an edge has weight $\geq \lfloor w(C)/2\rfloor - M\geq t-M$, then it is not removed. Hence again $(s^1,v_{i}^2)$ and $(v_{i+1}^3,s^1)$ are not removed because their weight is at least $t-M$ as follows from Lemma~\ref{lemma:middle}.

All remaining edges in $(V^3\times V^1)\cup (V^1\times V^2)$ have integral weights in $[t-M,t+M/2]$, and $C$ is still represented by the minimum triangle $\{ s^1,v_{i}^2,v_{i+1}^3\}$.
Now, for every remaining edge $(a,b)\in (V^1\times V^2) \cup  (V^3\times V^1)$, change its weight to $D[a,b]-t$.
The weights of the edges of $G'$ are now in the interval $[-M,M]$. Furthermore, since the weights of all triangles have decreased by $2t$, the minimum triangle of $G'$ is still the same. This completes the construction of $G'$.

\paragraph{Derandomization.} The only randomized part of our reduction is
our use of Yuster and Zwick's result.
Their algorithm can be derandomized, as pointed out in their paper~\cite{yusterzwick05} without affecting our use of their result.
Hence, we obtain a deterministic reduction from minimum cycle in directed
graphs to minimum triangle in undirected graphs which runs in $O(Mn^\omega\log n + n^\omega(\log Mn)\log n)$ time and does not increase the size of the graph or the edge weights by more than a constant factor.

\section{Discussion}\label{s-discuss}
We have obtained separate algorithms for minimum cycle for undirected graphs with nonnegative weights and for directed graphs with possibly negative weights. A natural question is whether one can obtain an algorithm that works for both types of graphs, or more generally for
 \emph{mixed} graphs: graphs with both directed and undirected edges. This turns out to be possible for mixed graphs with nonnegative weights. (The problem is NP-hard for mixed graphs with positive and negative weights, even when there are no negative cycles.)
 The idea for the proof of Theorem~\ref{thm:mix} below is to compute the distance estimates $D[\cdot,\cdot]$ using the approach from our reduction for directed graphs. This is possible since when the weights are nonnegative, one can reduce the shortest paths problem in undirected or mixed graphs to that in directed graphs by replacing each undirected edge $(u,v)$ by the two directed edges $(u,v)$ and $(v,u)$. Then the entire approach from the previous section applies up until the triangle instance needs to be constructed.
To construct the triangle instance, we use the color-coding technique from our minimum cycle algorithm for undirected graphs, but only on the undirected edges. The derandomization also follows from the previous two sections. More details follow.

Just as in the construction for directed graphs,
for every directed edge $(u,v)$ of $G$, add an edge from $u^2\in V^2$ to $v^3\in V^3$ with weight $w(u,v)$.
As in the construction for undirected graphs, randomly assign every node of $G$ one of two different colors $\{C_1,C_2\}$ independently uniformly at random. For every undirected edge $(u,v)$ of $G$, add an edge from $u^2\in V^2$ to $u^3\in V^3$ of weight $w(u,v)$ if and only if $\color(u)=\color(v)=C_2$.

Consider two nodes $x,y$ so that $D[x,y]<\infty$. Let $x'=\Pi[x,y]$ and $y'=\Pi'[x,y]$ be the first node after $x$ and the last node before $y$, respectively, on the path from $x$ to $y$ with weight $D[x,y]$. If $(x,x')$ is a directed edge, then add an edge from $x^3\in V^3$ to $y^1\in V^1$, just as in the directed graph construction. Similarly, if $(y',y)$ is a directed edge, then add an edge from $x^1\in V^1$ to $y^2\in V^2$. Otherwise, if $(y',y)$ is undirected,
add an edge from $x^1\in V^1$ to $y^2\in V^2$ only if $\color(y')=C_1$ and $\color(y)=C_2$. If $(x,x')$ is undirected, add an edge from $x^3\in V^3$ to $y^1\in V^1$ only if $\color(x')=C_1$ and $\color(x)=C_2$. The weights of these edges are all $D[x,y]$.

By Lemma~\ref{lemma:twopaths} (which also applies to mixed graphs) now every triangle in $G'$ corresponds to a simple cycle of no larger weight. Furthermore, the minimum cycle of $G$ is represented by a triangle of the same weight in $G'$ with constant probability. This follows from the color-coding and by the fact that for some node $s$ of $C$ and its middle edge $(v_{i},v_{i+1})$, $D[s,v_{i}]=d(s,v_{i})=d_C(s,v_{i})$ and $D[v_{i+1},s]=d(v_{i+1},s)=d_C(v_{i+1},s)$.

\begin{theorem}
Let $G(V,E,w)$ be a mixed graph on $n$ nodes, $w:E\rightarrow \{1,\ldots, M\}$.
In $\tilde{O}(Mn^\omega)$
time one can
construct $O(\log n)$ graphs $G'_1,\ldots,G'_k$ on $\Theta(n)$ nodes and edge weights in $\{1,\ldots,O(M)\}$ so that
the minimum out of all weights of triangles in the graphs $G'_i$
  is exactly the weighted girth of $G$.
\label{thm:mix}
\end{theorem}

Shortest paths and minimum cycles in undirected graphs with positive and negative weights are of a completely different nature than the corresponding problems in directed graphs or in undirected graphs with nonnegative weigths.
In the absence of negative cycles, shortest paths and cycles can be solved via matching techniques (see e.g.~\cite{lawler,kortevygen,gabow83,gabow85}).
However, the running times are not as good as the corresponding ones for directed graphs. For instance, APSP can be solved in $O(\min\{n^3,mn\log n\})$ time~\cite{gabow83}, whereas the corresponding problem in directed graphs can be solved in $O(\min\{n^3\log\log^3 n/\log^2 n,mn+n^2\log n\})$~\cite{chan07j,floyd,warshall}.

None of the shortest paths algorithms for directed graphs, including Yuster-Zwick's algorithm, apply for undirected graphs when there are negative weights, as they would confuse any negative weight edge with a negative cycle. Hence our approach from Section~\ref{s-directed} would not work.
Our approach from Section~\ref{s-undirected} also fails, even if we have already computed APSP in the graph. The main reason is that Lemma~\ref{lemma:twopaths} does not apply when the weights can be negative, and hence the color-coding technique cannot be applied, as is. Computing minimum cycles in undirected graphs with possibly negative weights may require entirely new techniques.

\paragraph{Extension to $k$-cycles.}
We now show that the minimum weight cycle problem in undirected and directed graphs can be reduced to the minimum $k$-cycle problem in undirected graphs for every $k\geq 4$.
In order to obtain our reduction to minimum $k$-cycle in undirected graphs with integral edge weights in $[0,O(M)]$ it suffices to provide a reduction from minimum weight triangle in an $n$-node undirected graph with weights in $[-M,M]$ to minimum $k$-cycle in a $\Theta(n)$-node undirected graph with edge weights in $[0,O(M)]$, and then to combine this reduction with our reduction from minimum weight cycle to minimum weight triangle. This proves Theorem~\ref{thm:equiv2}.

\begin{lemma} Let $k\geq 4$ be fixed.
Given an $n$-node undirected graph $G$ with integral edge weights in $[-M,M]$, one can construct in $O(n^2)$ time an undirected graph $G'$ on $\Theta(n)$ nodes and integral edge weights in $[0,6M]$ so that if $G$ has at least one triangle and the minimum triangle weight is $W$, then the minimum weight $k$-cycle in $G'$ has weight $W+15M$.
\end{lemma}

\begin{proof}
Without loss of generality, we can assume that the instance of minimum triangle is tripartite with partitions $V^1,V^2,V^3$.
Remove each node $v\in V^1$ and its incident edges and replace it with a path on $k-2$ nodes $P_v=\{v_1\rightarrow v_2\rightarrow\ldots\rightarrow v_{k-2}\}$ as follows.
For every original edge $(u,v)$ with $u\in V^2$, add an edge $(u,v_1)$ with weight $w(u,v)$, and for every original edge $(v,u)$ with $u\in V^3$, add an edge $(v_{k-2},u)$ with weight $w(v,u)$. Let the weights of the path edges $(v_i,v_{i+1})$ for $i\in \{1,\ldots,k-3\}$ be all $0$.
Increase the weights of all edges of $G'$ which are not on paths $P_v$ by $5M$.
This forms a weighted graph $G'$ on $O(kn)$ nodes and weights in $[0,6M]$.

Every triangle $v\in V^1,u\in V^2,z\in V^3$ of weight $W$ in the original graph has a corresponding $k-$cycle in $G'$ of weight $W+15M$.
Now consider any $k-$cycle $C$ of $G'$. If $C$ contains an edge of a path $P_v$ corresponding to a node $v\in V^1$, then it must contain the entire path since every node $v_i$ for $i\in \{2,\ldots,k-3\}$ has degree exactly $2$. Hence $C$ contains exactly $2$ other nodes which must close a cycle with $P_v$. Hence the other two nodes are from $V^2$ and $V^3$, and there is a corresponding triangle in $G$ of weight $15M$ less.

If on the other hand $C$ does not contain an edge of a path $P_v$, then it has $k$ edges of weight at least $4M$, and hence $w(C)\geq 4kM\geq 16M$ for $k\geq 4$.
Any triangle of $G$, however, corresponds to a cycle of weight $\leq 15M<16M\leq w(C)$. Hence the minimum weight $k$-cycle in $G'$ must correspond to a triangle in $G$, if $G$ contains a triangle.
\end{proof}


%
%
\begin{small}

\end{small}
\end{document}